\documentclass[12pt]{article}

\usepackage{hyperref}
\usepackage{amsmath, amssymb, amsthm}
\usepackage[sort&compress,numbers]{natbib}
\usepackage{doi}
\usepackage[margin=0.8in]{geometry}

\usepackage{amsmath, amssymb, amsthm, mathtools}

\newtheorem{lemma}{Lemma}
\newtheorem{theorem}{Theorem}
\newtheorem{coro}{Corollary}

\theoremstyle{definition}
\newtheorem{defi}{Definition}

\theoremstyle{remark}
\newtheorem{remark}{Remark}

\def\Ha{\mathcal H}
\def\Te{\mathcal T}
\def\Ka{\mathcal K}
\def \Tr{\mathrm{Tr}\,}
\def\Se {\mathcal S}
\def\supp{\mathrm{supp}}
\def\<{\langle\,}
\def\>{\,\rangle}

\title{Recoverability of quantum channels via hypothesis testing}
\author{Anna Jen\v cov\'a\footnote{Mathematical Institute, Slovak Academy of Sciences, \v
Stef\'anikova 49, 814 73 Bratislava, Slovakia, jenca@mat.savba.sk}}
\date{}
\begin{document}

\maketitle

\abstract{A quantum channel is sufficient with respect to a set of input states if it can
be reversed on this set. In the approximate version, the input states can be recovered
within an error bounded by the decrease of the relative entropy under the channel.  
Using a new integral representation of the relative entropy in
[arxiv:2208.12194], we present an easy proof of a characterization  of sufficient quantum channels
and recoverability by preservation of optimal success probabilities in hypothesis testing
problems, equivalently, by preservation of $L_1$-distance.}

\section{Introduction}

One of the fundamental properties of quantum relative entropy is monotonicity under
quantum channels, or the data processing inequality (DPI):
\begin{equation}\label{eq:dpi}
D(\Phi(\rho)\|\Phi(\sigma))\le D(\rho\|\sigma)
\end{equation}
for any pair of quantum states $\rho,\sigma$ and any completely positive  trace preserving
map, \cite{lindblad1975completely,uhlmann1977relative}. The DPI implies other important quantum
entropic inequalities, such as the Holevo bound \cite{holevo1973bounds}, strong subadditivity of
von Neumann entropy (SSA) \cite{lieb1973proof} or the  joint convexity of relative
entropy. In fact, SSA, joint convexity and DPI are all equivalent, see
\cite{ruskai2002inequalities} and the proof of DPI in 
\cite{lindblad1975completely} is based on the SSA.

The question when the data processing inequality becomes an equality for a completely
positive map and a pair of states was first
answered by Petz \cite{petz1986sufficient,petz1988sufficiency}, who proved that, provided the
relative entropy $D(\rho\|\sigma)$ is finite, equality occurs if and only if the two
states can be fully recovered. This means that there exists a channel $\Psi$ such that 
$\Psi\circ\Phi(\rho)=\rho$ and $\Psi\circ\Phi(\sigma)=\sigma$. In this case, we say that
the channel $\Phi$ is sufficient with respect to the pair of states
$\{\rho,\sigma\}$, in analogy with the classical notion of a statistic sufficient with
respect to a family of probability distributions. Moreover, Petz proved that there exists
a universal recovery channel $\Phi_\sigma$, such that
$\Phi_\sigma\circ\Phi(\sigma)=\sigma$ and we have $\Phi_\sigma\circ \Phi(\rho)=\rho$ if
and only if the channel $\Phi$ is sufficient with respect to $\{\rho,\sigma\}$. 

Sufficiency, or sometimes called reversibility, of channels was studied in a number of
subsequent works and several characterizations and applications were found,
\cite{mosonyi2004structure,jencova2006sufficiencyfirst, jencova2006sufficiency,
hiai2011quantum, jencova2012reversibility, 
shirokov2013reversibilityjmp, luczak2014quantum}. In particular, equality
conditions for various forms of DPI were studied, e.g.  \cite{ruskai2002inequalities,
jencova2010aunified, leditzky2017data}, and their relation to sufficiency were examined for other information
theoretic or statistical quantities, such as  different versions of quantum
$f$-divergences \cite{hiai2017different, hiai2021quantum}, quantum R\'enyi relative
entropies \cite{jencova2017preservation,jencova2018renyi,jencova2021renyi,hiai2021quantum}, Holevo quantity \cite{shirokov2013reversibility}, Fisher information
and $L_1$-distance \cite{jencova2012reversibility}.

An approximate version of sufficiency, called (approximate) recoverability is a much
stronger result stating that if the decrease in the relative entropy is small, there
exists a channel that recovers $\sigma$  perfectly while $\rho$ is recovered up to a small error. First result of this form was proved in
the work of Fawzi and Renner \cite{fawzi2015quantum}, who considered approximate quantum
Markov chains. This was soon extended to more general channels
\cite{wilde2015recoverability,sutter2016strengthened, sutter2017multivariate,
junge2018universal} and a variety of quantities such as $f$-divergences
\cite{carlen2018recovery,carlen2020recovery}, optimized $f$-divergences
\cite{gao2021recoverability} and Fisher information \cite{gao2023sufficient}. An important
result in this context is existence of an universal recovery channel $\Phi^u_\sigma$ depending only on the
state $\sigma$ such that \cite{junge2018universal}
\[
D(\rho\|\sigma)-D(\Phi(\rho)\|\Phi(\sigma))\ge -2\log F(\rho,
\Phi^u_\sigma\circ\Phi(\rho))\ge \|\rho-\Phi^u_\sigma\circ\Phi(\rho)\|_1^2,
\]
here $\|\cdot\|_1$ denotes the trace norm and
$F(\rho,\sigma)=\|\rho^{1/2}\sigma^{1/2}\|_1$ is the fidelity. See also
\cite{gao2021recoverability,faulkner2022approximate,faulkner2022approximate2} for the respective results for
normal states of von Neumann algebras.

In the simplest setting of quantum hypothesis testing, the null hypothesis $H_0=\sigma$ is tested
against the alternative $H_1=\rho$. The tests are represented by operators $0\le M\le I$,
with the interpretation that $\Tr[\omega M]$ is the probability of rejecting the
hypothesis if the true state is $\omega$.
For the  test represented by $M$, the  Bayes error probability for $\lambda\in [0,1]$ is
expressed as
\[
P_e(\lambda,\rho,\sigma,M)=\lambda\Tr[\sigma M]+(1-\lambda)\Tr[\rho(I-M)]
\]
and the test is Bayes optimal for $\lambda$ if this error probability is minimal over all
possible tests. It is quite clear that if we replace the states by $\Phi(\sigma)$ and
$\Phi(\rho)$, the achievable error probabilities cannot be decreased. It is a natural
question when the optimal error probabilities are preserved under $\Phi$, which is
equivalent to  preservation of the $L_1$-norm:
\begin{equation}\label{eq:l1}
\|\rho-s\sigma\|_1=\|\Phi(\rho)-s\Phi(\sigma)\|_1,\qquad \forall s.
\end{equation}
In classical
statistics, the theorem of Pfanzagl \cite{pfanzagl1974acharacterization,
strasser1985mathematical} states that if the achievable  error probabilities
for a pair of probability measures
$\{P_0,P_1\}$ do not increase after transformation by a statistic $T$, then $T$ must be sufficient with respect to
$\{P_0,P_1\}$. The corresponding result for quantum channels was investigated  in
\cite{jencova2010quantum, jencova2012reversibility},  and in \cite{luczak2015onageneral},
where more general risk functions for decision problems were considered. The equivalent
question of preservation of the $L_1$-distance, with applications to error correction,
was studied in \cite{blume2010information, ticozzi2010quantum}. In all these
works, additional conditions were needed, such as the equalities have to be assumed either for
larger sets of states with a special structure, or for any number of copies of $\rho$ and $\sigma$. 
The case when $\rho$ and $\sigma$ commute, or the channel $\Phi$ has commutative range,
was solved in \cite{jencova2012reversibility}.

Many of the results on recoverability of channels rely on an integral representation of
the relative entropies or  other  quantities in question such as $f$-divergences. These
 formulas are based on integral representation of operator convex functions.
Recently, a new integral formula for the relative entropy of positive semidefinite matrices
matrices was proved in
\cite{frenkel2022integral}. This formula can be easily extended to infinite dimensional
Hilbert spaces and rewritten in terms of the optimal
Bayes error probabilities. We use this formula for simple proofs of a
characterization of recoverability of quantum channels by preserving hypothesis testing
error probabilities, or equivalently $L_1$-distances, without any additional assumptions needed in the
previous works.

\section{Preliminaries}

Throughout this paper, $\Ha$ is a 
Hilbert space and we denote by $\Te(\Ha)$ the set of operators with finite trace and by $\Se(\Ha)$ the set of states (density operators) on $\Ha$,
that is, positive operators of trace 1. For a self-adjoint operator $A\in B(\Ha)$, $A_\pm$
denotes the positive/negative part of $A$ and for $A\ge 0$, we denote the projection
onto the support of $A$ by $\supp(A)$. The $L_1$-distance in $\Te(\Ha)$ is defined as
\[
\|S\|_1:=\sup_{\|A\|\le 1} \Tr[AS]=\Tr|S|,\qquad S\in \Te(\Ha).
\]
If $S\in \Te(\Ha)$ is self-adjoint, then we have
\[
\Tr[S_+]=\sup_{0\le M\le I} \Tr[MS],\quad \Tr[S_-]=-\inf_{0\le M\le I} \Tr[MS]
\]
and 
$\Tr[S_\pm]=\tfrac12(\|S\|_1\pm\Tr[S])$.

A quantum channel $\Phi$ is a completely
positive trace preserving map  $\Te(\Ha)\to \Te(\Ka)$.  The adjoint of $\Phi$ is the map $\Phi^*:B(\Ka)\to B(\Ha)$, defined by
\[
\Tr[\Phi^*(A)\rho]=\Tr[A\Phi(\rho)],\qquad A\in B(\Ha),\ \rho\in \Se(\Ha).
\]
It is easily seen that $\Phi^*$ is completely positive and unital. 

For positive operators  $\rho,\sigma \in \Te(\Ha)$, the quantum relative entropy is defined as
\[
D(\rho\|\sigma)=\begin{dcases} \Tr[\rho(\log(\rho)-\log(\sigma))], &  \supp(\rho)\le \supp(\sigma)\\
\infty, & \text{otherwise}.
\end{dcases}
\]
Relative entropy satisfies the data processing inequality \eqref{eq:dpi} which
holds for any pair of states $\rho,\sigma$ and any quantum channel $\Phi: \Te(\Ha)\to
\Te(\Ka)$.

\subsection{Quantum hypothesis testing and $L_1$-distance}

In the problem of hypothesis testing, the task is to test the hypothesis $H_0=\sigma$ against the
alternative $H_1=\rho$. Any test is represented by  an effect on $\Ha$, that is, an
operator $0\le M\le I$, corresponding to rejecting the hypothesis. For a test $M$, the error probabilities
are
\[
\alpha(M)=\Tr[\sigma M],\qquad \beta(M)= \Tr[\rho(I-M)].
\]
For $\lambda\in (0,1)$, we define the Bayes optimal test as the minimizer of 
\[
P_e(\lambda,\sigma,\rho,M):=\lambda\alpha(M)+(1-\lambda)\beta(M)=(1-\lambda)(1-\Tr[(\rho-s\sigma)M]),\qquad
s=\frac{\lambda}{1-\lambda}.
\]
The proof of the following description of the Bayes optimal tests can be found in
\cite{jencova2010quantum}.

\begin{lemma}[Quantum Neyman-Pearson] \label{lemma:qnm} Let $\rho,\sigma$ be states,
$\lambda \in [0,1]$ and put $s=\frac{\lambda}{1-\lambda}$. A test $M$ is a Bayes optimal
test for $\lambda,\sigma,\rho$ if and only if
\[
M=P_{s,+}+X_s,\qquad 0\le X_s\le P_{s,0},
\]
where $P_{s,\pm}=\supp((\rho-s\sigma)_\pm)$ and
$P_{s,0}=I-P_{s,+}-P_{s,-}$. The optimal error probability is then
\begin{align*}
P_e(\lambda,\sigma,\rho):=\max_M
P_e(\lambda,\sigma,\rho,M)&=(1-\lambda)(1-\Tr[(\rho-s\sigma)_+])\\
&=(1-\lambda)(s-\Tr[(\rho-s\sigma)_-])\\
&=\frac12(1-(1-\lambda)\|\rho-s\sigma\|_1).
\end{align*}

\end{lemma}

It is easily seen that the error probabilities and the related quantities in the above lemma
are monotone under channels, in particular,
\begin{align*}
P_e(\lambda,\Phi(\sigma),\Phi(\rho))&\ge P_e(\lambda,\sigma,\rho),\\
\|\Phi(\rho)-s\Phi(\sigma)\|_1&\le \|\rho-s\sigma\|_1,\\
\Tr[(\Phi(\rho)-s\Phi(\sigma))_-]&\le \Tr[(\rho-s\sigma)_-].
\end{align*}
In fact, monotonicity holds if $\Phi$ is a positive trace preserving map, so complete
positivity is not needed.

\subsection{Integral formula for the relative entropy}

The following new integral representation of the relative entropy was proved by Frenkel  in
\cite{frenkel2022integral}, in the case $\dim(\Ha)<\infty$.

\begin{theorem}\label{thm:frenkel} Let $\rho,\sigma$ be positive operators in $\Te(\Ha)$. Then
\[
D(\rho\|\sigma)=\Tr[\rho-\sigma]+\int_{-\infty}^\infty
\frac{dt}{|t|(1-t)^2}\Tr[((1-t)\rho+t\sigma)_-].
\]
\end{theorem}

\begin{proof} By \cite[Theorem 6]{frenkel2022integral}, the equality holds if
$\dim(\Ha)<\infty$. We will now prove that it can be extended to the case when $\dim(\Ha)=\infty$.
Assume first that $\supp(\rho)\le \supp(\sigma)$, so that we may assume that $\sigma$ is
faithful and $\Ha$ is separable by restriction to the support of $\sigma$. We will use a standard limiting
argument to extend the finite dimensional result to the separable case.

Let $P_n$ be an increasing sequence of finite rank projections such that $\vee_n P_n=I$.
Put $\rho_n=P_n\rho P_n$, $\sigma_n=P_n\sigma P_n$. Then restricted to the finite
dimensional space $P_n\Ha$, $\rho_n$ and $\sigma_n$ are positive semidefinite operators
with $\supp(\rho_n)\le \supp(\sigma_n)$.   Moreover, $\lim_n
\Tr[\rho_n]=\Tr[\rho]$, $\lim_n\Tr[\sigma_n]=\Tr[\sigma]$ and by \cite[Theorem 4.5]{hiai2018quantum} we have
$D(\rho\|\sigma)=\lim_n D(\rho_n\|\sigma_n)$.

For $t\in \mathbb R$ and $n\in \mathbb N$, put
\[
f_n(t):=\Tr[((1-t)\rho_n-t\sigma_n)_-], \qquad f(t):=\Tr[((1-t)\rho-t\sigma)_-].
\]
Then
\begin{align*}
f_n(t)&=\Tr[(t\sigma_n-(1-t)\rho_n)_+]=\sup_{0\le M_n\le P_n}
\Tr[M_n(t\sigma-(1-t)\rho)]\\ &\le \sup_{0\le M_{n+1}\le
P_{n+1}}\Tr[M_{n+1}(t\sigma-(1-t)\rho)]=\Tr[((1-t)\rho_{n+1}-t\sigma_{n+1})_-]=f_{n+1}(t),
\end{align*}
where the inequality follows from $0\le M_n\le P_n\le P_{n+1}$. Furthermore,  since
$P_n\to I$ in the strong operator topology, we have using \cite[Theorem 1]{grumm1973two} that 
 \[
\|P_n((1-t)\rho-t\sigma)P_n\|_1\to \|(1-t)\rho-t\sigma\|_1.
\]
It follows that
\begin{align*}
f_n(t)&=\tfrac12(\|(1-t)\rho_n-t\sigma_n\|_1-\Tr[(1-t)\rho_n-t\sigma_n])\to\tfrac12
(\|(1-t)\rho-t\sigma\|_1-\Tr[(1-t)\rho-t\sigma))\\
&=f(t).
\end{align*}
Hence $f_n$ is an increasing sequence of positive integrable functions converging
pointwise to $f$.  Since the integral formula holds in finite dimensions, we see using the Lebesgue monotone
convergence theorem that
\begin{align*}
D(\rho\|\sigma)&=\lim_n D(\rho_n\|\sigma_n)=\lim_n
\left(\Tr[\rho_n-\sigma_n]+\int_{-\infty}^\infty\frac{dt}{|t|(1-t)^2}f_n(t)\right)\\
&=\Tr[\rho-\sigma]+\int_{-\infty}^\infty
\frac{dt}{|t|(1-t)^2}f(t).
\end{align*}

If $\supp(\rho)\not\le\supp(\sigma)$, then there is some projection $Q$ such that $\Tr[\sigma
Q]=0$ and $c:=\Tr[\rho Q]>0$. Then for any $t>1$ we have 
\[
\Tr[((1-t)\rho-t\sigma)_-]\ge \Tr[Q(t\sigma-(1-t)\rho)]=(t-1)c
\]
and hence
\[
\int_{-\infty}^{\infty}\frac{dt}{|t|(t-1)^2}\Tr[((1-t)\rho-t\sigma)_-]\ge
c\int_1^\infty\frac{dt}{t(t-1)}=\infty.
\]
In this case we also have $D(\rho\|\sigma)=\infty$ by definition.

\end{proof}

The integral  formula leads to an easy proof of the fact that DPI for the relative entropy holds
for all positive trace preserving maps. This fact was first
proved in \cite{mullerhermes2017monotonicity}, using interpolation techniques.

\begin{coro}\label{coro:integral} Let $\rho,\sigma\in \Se(\Ha)$. Then for any $\lambda,\mu\ge 0$ such that 
$\mu\sigma\le \rho\le \lambda \sigma$, we have 
\[
D(\rho\|\sigma)=\int_\mu^\lambda \frac{ds}{s}\Tr[(\rho-s\sigma)_-]+ \log(\lambda)+1-\lambda.
\]

\end{coro}

\begin{proof} Since $((1-t)\rho+t\sigma)_-=0$
for $t\in [0,1]$, the integral splits into two parts, integrating over $t\le 0$ and $t\ge
1$. For the first integral, since $1-t>0$, we have
$((1-t)\rho+t\sigma)_-=(1-t)(\rho-\frac{t}{t-1}\sigma)_-$ and 
\begin{align*}
\int_{-\infty}^0\frac{dt}{-t(1-t)^2}\Tr[(1-t)\rho+t\sigma]_-&=\int_{-\infty}^0
\frac{dt}{t(t-1)}\Tr[(\rho-\frac{t}{t-1}\sigma)_-]=\int_0^1\frac{ds}{s}\Tr[(\rho-s\sigma)_-]\\
&=\int_\mu^1\frac{ds}{s}\Tr[(\rho-s\sigma)_-].
\end{align*}

For $t\ge 1$, we use
$((1-t)\rho+t\sigma)_-=((t-1)\rho-t\sigma)_+=(t-1)(\rho-\frac{t}{t-1}\sigma)_+$ and
inserting into the integral, we obtain
\[
\int_1^\infty\frac{dt}{t(t-1)^2}\Tr[((1-t)\rho+t\sigma)_-]=\int_1^\infty
\frac{ds}{s}\Tr[(\rho-s\sigma)_+]=\int_1^\lambda \frac{ds}{s}\Tr[(\rho-s\sigma)_+].
\]
The proof is finished by using  the equality $\Tr[(\rho-s\sigma)_+]=1-s +\Tr[(\rho-s\sigma)_-]$.
\end{proof}

\begin{remark}\label{rem:hilbert}
The smallest value of $\lambda$ in the above expression is related to the quantum max-relative entropy
defined as\[
D_{\max}(\rho\|\sigma):=\log\min\{\lambda,\ \rho\le \lambda\sigma\}.
\]
Similarly, the largest value of $\mu$ is $e^{-D_{\max}(\sigma\|\rho)}$. An important
related quantity is the Hilbert projective metric \cite{bushell1973hilbert}
\[
D_{\Omega}(\rho\|\sigma):=D_{\max}(\rho\|\sigma)+D_{\max}(\sigma\|\rho).
\]
See \cite{reeb2011hilbert,regula2022probabilistic, regula2022postselected} for more details and interpretations in the context of
quantum information theory. Note also that we may always put $\mu=0$ and if $\dim(\Ha)<\infty$, then the conditions that
$\rho\le \lambda\sigma$ for some $\lambda>0$ is equivalent to $\supp(\rho)\le
\supp(\sigma)$, so it holds whenever $D(\rho\|\sigma)$ is finite.
In infinite dimensions, this condition is much more restrictive. 

\end{remark}

\subsection{Sufficiency and recoverability for quantum channels}

The following definition first appeared in \cite{petz1988sufficiency} and can be seen as a
quantum generalization of the classical notion of a sufficient statistic.

\begin{defi}\label{def:sufficiency}
We say that a channel $\Phi:B(\Ha)\to B(\Ka)$ is sufficient with respect to a set of
states $\Se\subseteq \Se(\Ha)$ if there exists a channel $\Psi:B(\Ka)\to B(\Ha)$ such that
\[
\Psi\circ \Phi(\rho)=\rho,\qquad  \forall \rho\in \Se.
\]
\end{defi}

For a state $\sigma\in \Se(\Ha)$, we define an inner product $\<\cdot,\cdot\>_\sigma$ in
$B(\supp(\sigma))$ by
\[
\<A,B\>_\sigma:=\Tr[A^*\sigma^{1/2}B\sigma^{1/2}],\qquad A,B\in B(\supp(\sigma)).
\]
It was proved in \cite{petz1988sufficiency} that the (unique) linear map  $\Phi_\sigma: \Te(\supp(\Phi(\sigma)))\to
\Te(\supp(\sigma))$ determined by
\[
\<\Phi^*(B),A\>_\sigma=\<B,\Phi_\sigma^*(A)\>_{\Phi(\sigma)},\qquad A\in B(\supp(\sigma)),
B\in B(\supp(\Phi(\sigma)))
\]
is a channel, called the Petz dual of $\Phi$ with respect to $\sigma$ (or the Petz
recovery map). Note that we always have $\Phi_\sigma\circ
\Phi(\sigma)=\sigma$ and as it was further proved in \cite{petz1988sufficiency}, if both
$\sigma$ and $\Phi(\sigma)$ are faithful, then $\Phi$ is sufficient with respect to $\Se$
if and only if $\Phi_\sigma\circ \Phi(\rho)=\rho$ for all $\rho\in \Se$, so that
$\Phi_\sigma$ is a universal recovery channel.
\begin{remark} If $\dim(\Ha)<\infty$, we obtain the familiar form of the Petz recovery
channel:
\[
\Phi_\sigma(\cdot)=
\sigma^{1/2}\Phi^*(\Phi(\sigma)^{-1/2}\cdot \Phi(\sigma)^{-1/2})\sigma^{1/2}.
\]
\end{remark}

We also define 
\[
\Phi_{\sigma,t}(\cdot)= \sigma^{-it}\Phi_\sigma(\Phi(\sigma)^{it}\cdot
\Phi(\sigma)^{-it})\sigma^{it},\qquad t\in \mathbb R
\]
and
\[
\Phi_{\sigma,\mu}(\cdot)=\int_{-\infty}^\infty\Phi_{\sigma,t}(\cdot)d\mu(t),
\]
for a probability measure $\mu$ on $\mathbb R$. Clearly, all these maps are channels $\Te(\supp(\Phi(\sigma)))\to
\Te(\supp(\sigma))$ that recover the state $\sigma$.

%

\begin{theorem}\label{thm:petz}  Assume that $\rho,\sigma \in \Se(\Ha)$ are such that 
$D(\rho\|\sigma)<\infty$. Then the following are equivalent.
\begin{enumerate}
\item[(i)] $\Phi$ is sufficient with respect to $\{\rho,\sigma\}$;
\item[(ii)] $D(\Phi(\rho)\|\Phi(\sigma))=D(\rho\|\sigma)$;
\item[(iii)] $\Phi_{\sigma,t}\circ \Phi(\rho)=\rho$, for some $t\in
\mathbb R$;
\item[(iv)] $\Phi_{\sigma,t}\circ \Phi(\rho)=\rho$, for  all $t\in
\mathbb R$;
\item[(v)] $\Phi_{\sigma,\mu}\circ \Phi(\rho)=\rho$ for some probability measure $\mu$.

\end{enumerate}

\end{theorem}

\begin{proof} In finite dimensions the proof follows from  \cite[Theorem 3.3]{wilde2015recoverability}. 
The proof in the general case will be given in the Appendix.

\end{proof}

The following is an approximate version of sufficiency of channels, called recoverability
of $\Phi$.

\begin{theorem}\label{thm:recoverability} \cite{junge2018universal} Let $\sigma\in
\Se(\Ha)$. Then for any channel $\Phi: \Te(\Ha)\to \Te(\Ka)$ there exists a channel
$\Phi^u_\sigma: \Te(\Ka)\to \Te(\Ha)$ such that 
$\Phi^u_\sigma\circ\Phi(\sigma)=\sigma$ and for any $\rho\in \Se(\Ha)$ we have
\[
D(\rho\|\sigma)\ge D(\Phi(\rho)\|\Phi(\sigma))-2\log F(\rho,\Phi^u_\sigma\circ
\Phi(\rho))\ge D(\Phi(\rho)\|\Phi(\sigma))+ \frac14
\|\rho-\Phi^u_\sigma\circ\Phi(\rho)\|_1^2.
\]
\end{theorem}

In the above theorem, $F(\rho_0,\rho_1)$ is the fidelity
\[
F(\rho_0,\rho_1)=\|\rho_0^{1/2}\rho_1^{1/2}\|_1.
\]
The second inequality in Theorem \ref{thm:recoverability} is
obtained using the inequality $-\log(x)\ge 1-x$ for $x\in
(0,1)$ and the Fuchs-van de Graaf inequality, \cite{fuchs1999cryptographic}.

The universal recovery channel $\Phi^u_\sigma$ can be chosen as
\begin{equation}\label{eq:universal}
\Phi^u_\sigma(\cdot)=\Phi_{\sigma,\beta_0}(P\cdot P)+\Tr[(I-P)\cdot ],
\end{equation}
here $P=\supp(\Phi(\sigma))$ and $\beta_0$ is the probability density
function
\[
\beta_0(t)=\frac{\pi}{\cosh(2\pi t)+1}.
\]
Note that if $\supp(\rho)\le \supp(\sigma)$, then $\supp(\Phi(\rho))\le
\supp(\Phi(\sigma))$, so that
$\Phi^u_\sigma(\Phi(\rho))=\Phi_{\sigma,\beta_0}(\Phi(\rho))$ and the statement in this
case follows by \cite[Theorem 2.1]{junge2018universal}. If
$\supp(\rho)\not\le\supp(\sigma)$, then $D(\rho\|\sigma)=\infty$ and the inequality holds
trivially.

\section{Sufficiency and recoverability by hypothesis testing}

The characterization in Theorem \ref{thm:petz} and the integral formula in Corollary
\ref{coro:integral} now give an easy proof of characterization of sufficiency and
recoverability by quantities related to hypothesis testing. Note that here we do not have
to make any further assumptions about the states.

\begin{theorem}\label{thm:sufficiency} Let $\Phi:\Te(\Ha)\to \Te(\Ka)$ be a channel and let
$\rho,\sigma\in \Se(\Ha)$.
Then the following are
equivalent.
\begin{enumerate}
\item[(i)] $P_e(\lambda, \Phi(\rho),\Phi(\sigma))= P_e(\lambda,\rho,\sigma)$, for all  $\lambda\in
[0,1]$;
\item[(ii)]  $\|\Phi(\rho)-s\Phi(\sigma)\|_1= \|\rho-s\sigma\|_1$, for all
$s\ge 0$;
\item[(iii)]  $\Tr[(\Phi(\rho)-s\Phi(\sigma))_{+}]= \Tr[(\rho-s\sigma)_{+}]$, for all $s\ge 0$;

\item[(iv)]  $\Tr[(\Phi(\rho)-s\Phi(\sigma))_{-}]= \Tr[(\rho-s\sigma)_{-}]$, for all
$s\ge 0$;
\item[(v)] $\Phi$ is sufficient with respect to $\{\rho,\sigma\}$.

\end{enumerate}

\end{theorem}

\begin{proof} The equivalences between (i)-(iv) are clear from Lemma \ref{lemma:qnm}. 
 Assume that (iv) holds. Suppose first that $\rho\le \lambda\sigma$ for some
$\lambda>0$. Then also $\Phi(\rho)\le
\lambda\Phi(\sigma)$ and we have by Corollary \ref{coro:integral}
\begin{align*}
D(\Phi(\rho)\|\Phi(\sigma))&=\int_0^\lambda
\frac{ds}{s}\Tr[(\Phi(\rho)-s\Phi(\sigma))_-]+\log(\lambda)+1-\lambda\\
&=
\int_0^\lambda
\frac{ds}{s}\Tr[(\rho-s\sigma)_-]+\log(\lambda)+1-\lambda=D(\rho\|\sigma).
\end{align*}
 By Theorem
\ref{thm:petz}, this implies (v). In the general case, let
$\sigma_0=\frac12(\rho+\sigma)$, then $\rho\le 2\sigma_0$ and it is easily
seen that the equality (ii) implies a similar equality with $\sigma$ replaced by
$\sigma_0$. It follows that $\Phi$ is sufficient with respect to $\{\rho,\sigma_0\}$,
which implies (v).  The implication (v) $\implies$ (ii) follows from
monotonicity of the $L_1$-distance.

\end{proof}

We are now interested in a  similar result for recoverability. Assume first that there is a channel
$\Lambda: \Te(\Ka)\to \Te(\Ha)$ such that $\Lambda\circ \Phi(\sigma)=\sigma$ and $\|\Lambda\circ
\Phi(\rho)-\rho\|_1\le \epsilon$. We then have
\begin{equation}\label{eq:epsilon_norm}
\|\rho-s\sigma\|_1=\|\rho-\Lambda\circ\Phi(\rho)+\Lambda\circ\Phi(\rho)-s\Lambda\circ
\Phi(\sigma)\|_1\le \|\Phi(\rho)-s\Phi(\sigma)\|_1+\epsilon.
\end{equation}
Using Lemma \ref{lemma:qnm}, we see that the resulting inequality in \eqref{eq:epsilon_norm} is equivalent to any of
the following inequalities
\begin{align}
\Tr[(\Phi(\rho)-s\Phi(\sigma))_+]&\ge \Tr[(\rho-s\sigma)_+]-\frac{\epsilon}2,\qquad s\ge 0\\
\Tr[(\Phi(\rho)-s\Phi(\sigma))_-]&\ge
\Tr[(\rho-s\sigma)_-]-\frac{\epsilon}2\label{eq:epsilon_minus},\qquad s\ge 0\\
P_e(\lambda,\Phi(\sigma),\Phi(\rho))&\le
P_e(\lambda,\sigma,\rho)+\frac{1-\lambda}2\epsilon,\qquad \lambda\in [0,1].
\end{align}

The following result gives the converse statement. Note that here we will need the
assumption that the Hilbert projective metric $D_\Omega(\rho\|\sigma)$ is finite,
equivalently, that $\mu\sigma\le \rho\le
\lambda\sigma$ for some $\mu,\lambda>0$ (see Remark \ref{rem:hilbert}), to get a nontrivial
result.

\begin{theorem}\label{thm:recoverability_e}  Let $\rho,\sigma\in \Se(\Ha)$ and let
$\Phi:\Te(\Ha)\to \Te(\Ka)$ be a quantum  channel. If 
\[
\|\Phi(\rho)-s\Phi(\sigma)\|_1\ge \|\rho-s\sigma\|_1-\epsilon,\qquad \forall s\ge 0
\]
holds for some $\epsilon\ge 0$, then there exists a channel $\Lambda:\Te(\Ka)\to \Te(\Ha)$ such
that $\Lambda\circ \Phi(\sigma)=\sigma$ and 
\[
\|\Lambda\circ \Phi(\rho)-\rho\|_1\le
\sqrt{2\epsilon}D_\Omega(\rho\|\sigma)^{1/2}.
\]
Moreover, we may take $\Lambda =\Phi^u_\sigma$ as in \eqref{eq:universal}.

\end{theorem}

\begin{proof} The statement is trivial if $D_\Omega(\rho\|\sigma)=\infty$, so assume that 
 $\mu\sigma\le \rho\le \lambda\sigma$ for $\mu,\lambda>0$,
 $\mu=e^{-D_{\max}(\sigma\|\rho)}$ and $\lambda=e^{D_{\max}(\rho\|\sigma)}$. Then also $\mu \Phi(\sigma)\le
\Phi(\rho)\le \lambda\Phi(\sigma)$. By the assumptions, inequality \eqref{eq:epsilon_minus} 
holds. Using  Corollary \ref{coro:integral}, we get
\begin{align*}
D(\rho\|\sigma)-D(\Phi(\rho)\|\Phi(\sigma))&=\int_\mu^\lambda\frac{ds}{s}\bigl(\Tr[(\rho-s\sigma)_-]-\Tr[(\Phi(\rho)-s\Phi(\sigma))_-]\bigr)\\
&\le \frac{\epsilon}2
\int_\mu^\lambda\frac{1}{s}ds=\frac{\epsilon}2\bigl(\log(\lambda)-\log(\mu)\bigr)=\frac{\epsilon}2D_\Omega(\rho\|\sigma).
\end{align*}
The statement now
follows by Theorem \ref{thm:recoverability}.

\end{proof}

\begin{remark} The recoverability result can be also formulated in the setting of
comparison of statistical experiments, which is an extension of the classical theory of
Blackwell \cite{blackwell1953comparison}, T\"orgersen \cite{torgersen1991comparison} and
Le Cam \cite{lecam1964sufficiency}. A
(quantum)
statistical experiment is any parametrized family of (quantum) states. For two experiments
$\mathcal E$ and $\mathcal E_0$ with the same parameter set (not necessarily living on the
same Hilbert space), 
we say that $\mathcal E_0$ is $(2,\epsilon)$-deficient with respect to $\mathcal E$ if the error
probabilities of testing problems involving elements of $\mathcal E_0$ are up to $\epsilon$ not
worse than those of corresponding testing problems for $\mathcal E$. See  \cite{jencova2012comparison} for a more precise definition. In particular, for
$\mathcal E_0=\{\rho_0,\sigma_0\}$ and $\mathcal E=\{\rho,\sigma\}$, this amounts to the
condition 
\[
P_e(\lambda,\rho_0,\sigma_0)\le P_e(\lambda,\rho,\sigma)+\epsilon,\qquad \forall
\lambda\in [0,1].
\]
Using Lemma \ref{lemma:qnm}, we see that this is equivalent to any of the inequalities
\begin{align*}
\|\rho_0-s\sigma_0\|_1&\ge \|\rho-s\sigma\|_1-2\epsilon(1+s),\qquad s\ge 0\\
\Tr[(\rho_0-s\sigma_0)_{+}]&\ge \Tr[(\rho-s\sigma)_{+}]-\epsilon(1+s),\qquad s\ge
0\\
\Tr[(\rho_0-s\sigma_0)_{-}]&\ge \Tr[(\rho-s\sigma)_{-}]-\epsilon(1+s), \qquad s\ge 0.
\end{align*}
It is easily seen that this is true if there is some channel $\Lambda$ such that
$\|\Lambda(\rho_0)-\rho\|_1\le\epsilon$ and $\|\Lambda(\sigma_0)-\sigma\|_1\le \epsilon$.
In the classical case the converse holds, but  in the quantum case this is not true. 
We can obtain some form of the converse statement if $\rho_0=\Phi(\rho)$ and
$\sigma_0=\Phi(\sigma)$, in a similar way as in Theorem \ref{thm:recoverability_e}.

\end{remark}

\section*{Acknowledgement}

I am grateful to Mark Wilde and Fumio Hiai for their comments on an earlier version of the
manuscript.
The research was supported by the grant VEGA 1/0142/20 and  the Slovak
Research and Development Agency grant APVV-20-0069.

\appendix

\section{Appendix: Proof of Theorem \ref{thm:petz}}

The implication (i) $\implies$ (ii) follows by the data processing inequality
\eqref{eq:dpi}, (ii) $\implies$ (iii) (with $t=0$) was proved in
\cite{jencova2006sufficiency, petz1988sufficiency}. The implications (iii) $\implies$ (i)
and (iv) $\implies$ (v) $\implies$ (i) are easy, so the only thing left to prove is (i)
$\implies$ (iv). 

By the assumption $D(\rho\|\sigma)<\infty$, we have that $\supp(\rho)\le \supp(\sigma)$
and also $\supp(\Phi(\rho))\le \supp(\Phi(\sigma)$. We may therefore assume that both
$\sigma$ and $\Phi(\sigma)$ are faithful, by restriction to the respective supports.

With this assumption, we will also need to recall some further properties of sufficient channels from
\cite{petz1988sufficiency,jencova2006sufficiency}. Let us denote by $u_t$ and $v_t$ the Connes cocycles
\[
u_t=\rho^{it}\sigma^{-it},\quad v_t=\Phi(\rho)^{it}\Phi(\sigma)^{-it},\qquad t\in \mathbb
R,
\]
then $u_t$ and $v_t$ are one-parameter families of isometries satisfying the conditions
\[
\sigma^{is}u_t\sigma^{-is}=u_s^*u_{t+s},\quad
\Phi(\sigma)^{is}v_t\Phi(\sigma)^{-is}=v_s^*v_{t+s},\qquad s,t\in \mathbb R.
\]
By the results of \cite{petz1988sufficiency,jencova2006sufficiency}, we can see that if
$\Phi$ is sufficient with respect to $\{\rho,\sigma\}$, then $u_t$ is in the
multiplicative domain of the unital completely positive map $\Phi_\sigma^*$ and similarly $v_t$ is in the multiplicative domain of
$\Phi^*$, see \cite[Theorem 3.18]{paulsen2002completely} for the definition and properties of
multiplicative domains. Moreover, $\Phi^*(v_t)=u_t$ and $\Phi_\sigma^*(u_t)=v_t$. It
follows that for any $s,t\in \mathbb R$,
\[
\Phi^*(\Phi(\sigma)^{it}v_s\Phi(\sigma)^{-it})=\Phi^*(v_t^*v_{s+t})=u_t^*u_{s+t}=\sigma^{it}u_s\sigma^{-it}
\]
and similarly
$\Phi_\sigma^*(\sigma^{it}u_s\sigma^{-it})=\Phi(\sigma)^{it}v_s\Phi(\sigma)^{-it}$. 
Let now $A\in B(\Ha)$, we have for all $s,t\in \mathbb R$,
\[
\<\sigma^{it}u_sA^*\sigma^{-it},\sigma^{it}u_s\sigma^{-it}\>_\sigma=\<u_sA^*,u_s\>_\sigma=\Tr[A\rho^{is}\sigma^{1/2-is}\rho^{is}_s\sigma^{1/2-is}]
\]
and the  analytic continuation to $s=-\tfrac12i$ of the last expression becomes
$\Tr[A\rho]$. On the other hand, using the above properties of the cocycles and of
multiplicative domains, we get by the definition of the Petz dual
\begin{align*}
\<\sigma^{it}u_sA^*\sigma^{-it},\sigma^{it}u_s\sigma^{-it}\>_\sigma&=\<\sigma^{it}u_sA^*\sigma^{-it},\Phi^*(\Phi(\sigma)^{it}v_s\Phi(\sigma)^{-it})\>_\sigma\\
&=
\<\Phi_\sigma^*(\sigma^{it}u_sA^*\sigma^{-it}),\Phi(\sigma)^{it}v_s\Phi(\sigma)^{-it}\>_{\Phi(\sigma)}\\
&=\<\Phi(\sigma)^{it}v_s\Phi(\sigma)^{-it}\Phi_\sigma^*(\sigma^{it}A^*\sigma^{-it}),\Phi(\sigma)^{it}v_s\Phi(\sigma)^{-it}\>_{\Phi(\sigma)}\\
&=\Tr[\Phi^*_{t,\sigma}(A)\Phi(\rho)^{is}\Phi(\sigma)^{1/2-is}\Phi(\rho)^{is}\Phi(\sigma)^{1/2-is}],
\end{align*}
here the analytic continuation to $s=-\tfrac12i$ equals 
\[
\Tr[\Phi^*_{t,\sigma}(A)\Phi(\rho)]=\Tr[A\Phi_{t,\sigma}(\Phi(\rho))].
\]
Since $A$ and $t$ were arbitrary, this finishes the proof of (i) $\implies$ (iv).


\end{document}